\documentclass[12pt,a4paper,final]{article}

\usepackage[utf8]{inputenc}
\usepackage[T1]{fontenc}
\usepackage{amsmath,amssymb,amsfonts,amsthm}
\usepackage{geometry}
\usepackage{setspace}
\usepackage{natbib}
\usepackage{algorithm}
\usepackage{algpseudocode}
\usepackage{bm}
\usepackage{graphicx}
\usepackage[hidelinks]{hyperref}
\usepackage{float}
\usepackage{booktabs}
\usepackage{multirow}

\theoremstyle{plain}
\newtheorem{proposition}{Proposition}
\newtheorem{definition}{Definition}
\newtheorem{lemma}{Lemma}
\newtheorem{remark}{Remark}
\newtheorem{example}{Example}

\newcommand{\E}{E}
\newcommand{\R}{\mathbb{R}}
\newcommand{\bx}{\bm{x}}
\newcommand{\bX}{\bm{X}}
\newcommand{\by}{\bm{y}}
\newcommand{\bZ}{\bm{Z}}

\title{Distributionally balanced sampling designs via tactical configurations}
\author{
    Anton Grafström\thanks{Email: anton.grafstrom@slu.se, https://orcid.org/0000-0002-4345-4024} 
    \ and 
    Wilmer Prentius\thanks{Email: wilmer.prentius@slu.se, https://orcid.org/0000-0002-3561-290X} \\[0.3cm]
    \normalsize Department of Forest Resource Management, \\
    \normalsize Swedish University of Agricultural Sciences, Umeå, Sweden
}
\date{September 24, 2026}

\begin{document}
\maketitle

\begin{abstract}
\noindent 
Distributionally balanced sampling designs are low-discrepancy probability designs obtained by minimizing the expected discrepancy between the auxiliary-variable distribution of a random sample and the target population distribution. Existing constructions rely on circular population sequences, which restrict the design space by forcing samples to be contiguous blocks of a sequence. We propose a new construction based on tactical configurations that removes this topological constraint. The resulting designs are fixed-size and have equal inclusion probabilities. We develop both a simple initialization valid for arbitrary population and sample sizes and a spatial initialization that yields a lower initial expected discrepancy, together with a simulated annealing algorithm for optimization within this class. In simulations and empirical examples, the proposed method outperforms state-of-the-art alternatives in terms of distributional fit and balance.


\vskip.2cm
\noindent \textbf{Note:} This paper is accepted for publication in Survey Methodology and is expected to appear in the June 2027 issue.
\end{abstract}

\newpage

\section{Introduction}

Auxiliary information is often available for all population units before sampling and can then be incorporated into the sampling design to improve the precision of estimators. A classical approach is balanced sampling, where samples are selected so that Horvitz-Thompson estimates of prescribed auxiliary totals are equal or close to their known population totals. The cube method provides a general procedure for selecting balanced samples \citep{DevilleTille2004}. A more general design principle is spatially balanced sampling, which aims to spread the sample over the space defined by the auxiliary variables. As shown by \citet{GrafstromLundstrom2013}, well-spread samples are also approximately balanced. The local pivotal method (LPM), for example, discourages the joint selection of nearby units and thereby produces samples that cover the auxiliary space well \citep{GrafstromEtAl2012}. Such general design principles are particularly useful in multipurpose surveys and inventories, including applications such as environmental and natural-resource monitoring, where several target variables may be measured on the same sample.

\citet{Bea2015} illustrates the use of spatially balanced sampling in agricultural surveys, while \citet{Bea2017} provides a review of spatially balanced sampling methods. Several strategies have been proposed to achieve a well-spread sample. For example, \citet{BP2017} presents a design in which the probability of selecting a sample is proportional to the within sample distance. Other approaches include ordering units according to an approximate solution to the traveling salesman problem and then applying systematic sampling \citep{DicksonTille2016}, or using weakly associated vectors \citep{JauslinTille2020}. \citet{Rea2018} use the structural properties of the Halton sequence to partition the auxiliary space and select a well-spread sample. 

Balanced and spatially balanced sampling exploit different aspects of the available auxiliary information. Balancing controls specified functions of the auxiliary variables, typically their totals or means, whereas spatial balance aims to obtain broad coverage of the auxiliary space. Distributionally Balanced Designs (DBD), introduced by \citet{GrafstromPrentiusDBD}, take a more direct approach by seeking probability samples whose empirical auxiliary-variable distribution resembles the corresponding population distribution. Rather than targeting only a prescribed collection of auxiliary totals, the objective is to reproduce the auxiliary distribution more broadly. This is attractive when the relationships between the auxiliary variables and the target variables are not fully known at the design stage.

The energy distance \citep{SzekelyRizzo2013} measures the discrepancy between distributions and was first introduced in the sampling context by \citet{GrafstromPrentiusDBD}. For a given sample, the energy distance compares the empirical distribution of its auxiliary vectors with the empirical auxiliary distribution of the entire population. Unlike criteria based only on differences in means, the energy distance is sensitive to more general differences between multivariate distributions. In a finite population, it can be computed directly from pairwise distances between auxiliary vectors. A distributionally balanced sampling design can therefore be constructed by seeking a probability distribution over samples for which the expected energy distance between the random sample and the population is small.

The expected energy distance is relevant not only as a measure of distributional representativeness but also for estimation efficiency. \citet{GrafstromPrentiusDBD} show that it controls an upper bound on the design variance of the Horvitz-Thompson (HT) estimator for a broad class of target variables that vary smoothly with the auxiliary variables. Consequently, reducing the expected energy distance provides protection against high sampling variance for a broad class of relationships between auxiliary and target variables. This connection is central to the DBD principle: auxiliary information is used to construct a design whose samples reproduce the population structure relevant to a range of potential target variables, rather than to a single prespecified model.

The construction proposed by \citet{GrafstromPrentiusDBD}, hereafter referred to as circular DBD, represents the population by a circular sequence and defines the sampling design through contiguous blocks of that sequence. A random starting position determines which block is selected. This construction produces fixed-size samples with equal first-order inclusion probabilities, and the circular sequence can be optimized so that its contiguous blocks have small energy distance to the population auxiliary distribution.

The circular representation, however, imposes a substantial structural restriction on the design. Every admissible sample must consist of neighboring units in the same circular ordering. 
Thus, each unit, together with its neighbors, must constitute a distributionally representative sample, meaning that improvements in one part of the sequence may constrain what can be achieved elsewhere.
More fundamentally, the requirement that all samples arise as contiguous blocks of a single ordering is a property of the circular construction, not of the distributional-balance objective itself. This motivates searching for distributionally balanced designs over a broader class of fixed-size, equal-probability sampling designs.

In this paper, we obtain such a class using \emph{tactical configurations}, also known as binary $1$-designs  \citep[][Ch.~1]{Bea99}. In sampling terms, a tactical configuration is a finite collection of fixed-size samples in which every population unit occurs in the same number of samples. Therefore, selecting one of the samples uniformly gives equal first-order inclusion probabilities. We focus on \emph{minimum tactical configurations}, for which the number of samples in the configuration is the smallest compatible with the population size and sample size. Thus, the proposed construction retains a compact representation of the sampling design while removing the requirement that its samples be generated by a single population ordering.

We call the resulting framework Distributionally Balanced Designs via Tactical Configurations (DBD-TC). Instead of being restricted to a circular ordering, DBD-TC directly optimizes a tactical configuration with respect to the expected energy distance.

The remainder of the paper is organized as follows. Section~\ref{sec:framework} introduces the methodological foundations, including tactical configurations, their minimum size, and the energy-distance criterion. Section~\ref{sec:proposed} presents the proposed DBD-TC methodology, including initialization procedures and a stochastic optimization algorithm. Section~\ref{sec:simulation} evaluates DBD-TC in simulation studies and an empirical application, with particular attention to distributional balance and the efficiency of HT estimation. Section~\ref{sec:remarks} provides final remarks.

\section{Theoretical foundations} \label{sec:framework}

Let $U = \{1,2,\dots,N\}$ denote a finite population of size $N$, where each unit $i \in U$ is associated with a vector of auxiliary variables $\bx_i \in \R^p$. We seek a probability sampling design with a fixed sample size $n$ and equal first-order inclusion probabilities $\pi_i = n/N$ for all $i \in U$. First, we define the combinatorial structure from which samples are drawn.

\begin{definition}[Tactical configuration]
A design matrix $\bm{D}_{N \times M} = [\bm{d}_1 | \bm{d}_2 |\cdots| \bm{d}_M]$ of samples $\bm{d}_k \in \{0,1\}^N$ is a \emph{tactical configuration} if:
\begin{enumerate}
    \item $\sum_{i=1}^N d_{ik} = n > 0$ for all $k=1,\dots,M$;
    \item $\sum_{k=1}^M d_{ik} = c > 0$ for all $i\in U$,
\end{enumerate}
i.e. all $M$ samples have size $n$, and all population units appear in exactly $c$ samples. We call $M$, $c$ tactical configuration parameters.
\end{definition}

Based on this structure, we define the corresponding sampling design.

\begin{definition}[Tactical configuration sampling design]
A sampling design is a \emph{tactical configuration design} if it is constructed by selecting a column of $\bm{D}$ uniformly at random. 
Because $\bm{D}$ may contain duplicate samples, the probability of selecting a unique sample $\bm{d}$ is given by
\[p(\bm{d}) = P(\bm{\delta} = \bm{d}) = \frac{m(\bm{d})}{M},\]
where $\bm{\delta}$ is the random sample and $m(\bm{d})$ is the multiplicity of $\bm{d}$ in $\bm{D}$.
\end{definition}

By definition, in a tactical configuration design, the first- and second-order inclusion probabilities are exactly
\[
\pi_i = \frac{\sum_{k=1}^M d_{ik}}{M}
,\quad
\pi_{ij} = \frac{\sum_{k=1}^M d_{ik} d_{jk}}{M}
.
\]

Throughout the paper, the column vectors $\bm{d}_k$ of a tactical configuration $\bm{D}$ are referred to as samples, and $d_{ij}$ refers to elements of $\bm{D}$, while $d_i$ refers to the $i$-th element of a specific sample vector $\bm{d}$.

\subsection{Minimum tactical configurations}

In general, a tactical configuration may contain duplicate samples (when $c>1$). 
This implies that the size of the design's support (the set of unique samples $\mathcal{S}$ with strictly positive probability) is bounded by the size of the configuration, i.e., $|\mathcal{S}| \leq M$. For computational and practical efficiency, we aim to keep the configuration as small as possible.

\begin{definition}[Minimum tactical configuration]
A tactical configuration $\bm{D}$ is minimum if its size $M$, and thus the multiplicity constant $c$, are minimized. 
\end{definition}

The minimum possible size is determined by the greatest common divisor of $N$ and $n$.

\begin{proposition}[Theoretical bound on configuration size] \label{prop:min_support}
Let $g = \gcd(N,n)$. The size of a tactical configuration must be at least $M = N/g$, with each unit appearing exactly $c = n/g$ times.
\end{proposition}

\begin{proof}
In the $N \times M$ configuration $\bm{D}$, each column sums to $n$, and each row sums to $c$. Hence $Mn = Nc$, which rearranges to $M/c = N/n$. 
To minimize $M$ and $c$ subject to the constraint that both must be positive integers, we divide $N$ and $n$ by their greatest common divisor $g$. 
This yields the absolute minimum integer values $c = n/g$ and $M = N/g$.
\end{proof}

Circular DBD has support size $N$ via a sequence-based construction, whereas  Proposition \ref{prop:min_support} shows that a minimum tactical configuration has size $N/g$. Therefore, the support size of DBD-TC with a minimum configuration size is at most $N/g$, which is strictly smaller than $N$ when $g>1$, without imposing ordering constraints.

\subsection{Energy distance in sampling}

Let $F_{\bm{d}}$ and $F_U$ be the empirical distributions of $\{\bx_i : d_i = 1\}$ and $\{\bx_i:i\in U\}$, respectively. 
Also, let $\bX, \tilde{\bX}$ be independent random vectors distributed according to $F_{\bm{d}}$, 
and let $\bZ, \tilde{\bZ}$ be independent random vectors distributed according to $F_U$. 
The energy distance between these distributions is defined as
\begin{equation} \label{eq:ed}
\mathcal{E}(F_{\bm{d}},F_U)
=
2\E\|\bX-\bZ\|
-
\E\|\bX-\tilde{\bX}\|
-
\E\|\bZ-\tilde{\bZ}\|,
\end{equation}
where $\|\cdot\|$ denotes the Euclidean distance. The energy distance is zero if and only if $F_{\bm{d}} = F_U$ \citep{SzekelyRizzo2013}.
For a specific sample $\bm{d}$ of DBD-TC, the computable energy distance derived from \eqref{eq:ed} becomes
\begin{equation}\label{eq:energy-sample}
\mathcal{E}(F_{\bm{d}},F_U)
=
\sum_{i \in U} \Phi_i \left(2\frac{d_i}{n} - \frac{1}{N}\right)
-
\sum_{i \in U} \sum_{j \in U} \frac{d_i}{n}\frac{d_j}{n} \|\bx_i-\bx_j\|
,
\end{equation}
where
\[
\Phi_i = \frac{1}{N}\sum_{k\in U}\|\bx_i-\bx_k\|,
\quad i \in U.
\]

The use of expected energy distance to quantify the distributional discrepancy (imbalance) of a sampling design was introduced by \citet{GrafstromPrentiusDBD}. Moreover, for a fixed-size sampling design and a target variable that varies smoothly with respect to the auxiliary variables, the expected energy distance controls an upper bound on the variance of the expansion estimator 
 \citep[][Proposition 1]{GrafstromPrentiusDBD}. In short, their result says that for a target variable $y_i$ belonging to a broad class of smooth functions of $\bx_i$, with $Y=\sum_{i\in U}y_i$, the mean square error (MSE) of $\widehat{Y}_{\bm{\delta}}=(N/n)\sum_{i\in U}y_i\delta_i$ has the bound $\E\!\left[(\widehat{Y}_{\bm{\delta}}-Y)^2\right]
\;\le\;C\;
\E\!\left[\mathcal{E}(F_{\bm{\delta}},F_U)\right]$, where the expectation is under the design and $C$ is a constant.

\section{The proposed DBD-TC design} \label{sec:proposed}

Let $\mathcal{P}$ denote the set of all tactical configurations $\bm{D}$ of size $M$. In this paper, we use the minimum $M$ of Proposition~\ref{prop:min_support}. 
The DBD-TC objective is to find a tactical configuration $\bm{D}^{\ast}$ by minimizing the expected energy distance across the configuration:
\[
    \bm{D}^{\ast} \in \underset{\bm{D} \in \mathcal{P}}{\arg\min}
    \quad
    E\left[\mathcal{E}(F_{\bm{\delta}},F_U)\right],
\]
where 
\[
E\left[\mathcal{E}(F_{\bm{\delta}},F_U)\right]
=\bar{\mathcal{E}}(\bm{D})
=\frac{1}{M}\sum_{k=1}^M \mathcal{E}(F_{\bm{d}_k},F_U).
\]
Thus, the target of optimization is the configuration itself. 
Due to the large size of $\mathcal{P}$, we use simulated annealing to find a configuration $\bm{D}^{\circ}$ that produces an expected energy distance close to that of $\bm{D}^{\ast}$, thus reducing the MSEs for estimators of target variables related to the auxiliaries.

\begin{remark}[Configuration size]
Let $M_0$ denote the minimum tactical configuration size,
and let $\mathcal{P}_M$ denote the set of tactical configurations of size $M$. Any feasible configuration size is of the form $M=qM_0$ for some positive integer $q$. Moreover, if $\bm{D}\in\mathcal{P}_{M_0}$,
then repeating each column of $\bm{D}$ $q$ times yields a configuration
in $\mathcal{P}_{qM_0}$ with the same expected energy distance.
It follows that
\[
\min_{\bm{D}\in\mathcal{P}_{qM_0}}
\bar{\mathcal{E}}(\bm{D})
\leq
\min_{\bm{D}\in\mathcal{P}_{M_0}}
\bar{\mathcal{E}}(\bm{D}).
\]
Thus, restricting attention to minimum tactical configurations is motivated by computational efficiency rather than by optimality. Increasing the configuration size from its minimum feasible value cannot worsen the globally optimal expected energy distance and may improve it by allowing greater flexibility in the configuration.
\end{remark}

\subsection{A simple constructive initialization} 

To perform the optimization, we first require a valid starting point $\bm{D} \in \mathcal{P}$. Constructing a tactical configuration where every sample has size $n$ and every unit appears exactly $c$ times corresponds to a combinatorial design problem. While the space of such designs is large, a valid instance can be generated efficiently by mapping a permutation of the population onto the indices of the sample sets using modular arithmetic. The following randomized cyclic assignment method guarantees a valid initial configuration for any pair $N, n$.

\begin{lemma}
Let $\bm{v} \in \{0,1\}^M$ be a vector such that exactly $c$ entries are 1.
Then a minimum tactical configuration $\bm{D}$ is the matrix whose rows are the first $N$ cyclic shifts of $\bm{v}^\intercal$.
\end{lemma}

\begin{proof}
Assume that $\bm{D} \in \{0,1\}^{N \times M}$ has the size of a minimum tactical configuration, i.e. $M=N/g$ and $c=n/g$.
By the definition of $\bm{v}$, the rows of $\bm{D}$ sum to $c$.
A full cycle of $\bm{v}$ is completed after $M$ rows, and after $g$ rotations are completed,
the columns of $\bm{D}$ sum to $cg=n$.
\end{proof}

\begin{example}[Constructive initialization]
    For a population of size $N=6$ and a sample of size $n=4$, the minimum tactical configuration has parameters $M=3, c=2$, as $g=2$.
    An arbitrary cyclic vector is
    \[
    \bm{v} = [1,1,0]^\intercal .
    \]
    By using this vector as the initial row, and applying right cyclic shifts to generate the remaining rows, we obtain the minimum tactical configuration matrix
    \[
    \bm{D} = \begin{bmatrix}
        1 & 1 & 0 \\
        0 & 1 & 1 \\
        1 & 0 & 1 \\
        1 & 1 & 0 \\
        0 & 1 & 1 \\
        1 & 0 & 1 \\
    \end{bmatrix}
    .
    \]
\end{example}

If $\bm{D}$ is a minimum tactical configuration, then $\left[\bm{D} | \cdots | \bm{D}\right]$ is a tactical configuration. Since every feasible tactical configuration size is of this form, this provides a valid constructive initialization for any feasible configuration size. Furthermore, if $\bm{D}$ is a tactical configuration and $\bm{P}$ is a (random) permutation matrix of size $N\times N$, then $\bm{P}\bm{D}$ is also a tactical configuration.

\subsection{Sampling-based initialization}
While simple initialization ensures a valid tactical configuration, it typically results in a high initial distributional discrepancy.
To address this, we propose a sampling-based initialization. For a vector $\bm{p}$ of inclusion probabilities satisfying $\sum_{i\in U}p_i=n$, we use a randomized fixed-size sample generator
\[
f : [0,1]^N \rightarrow 
\left\{ \bm{d} \in \{0,1\}^N : \sum_{i \in U} d_i = n \right\},
\]
which has the property $E[f(\bm{p})]=\bm{p}$. Thus, $f$ generates a binary inclusion vector with exactly $n$ included units and with prescribed inclusion probabilities $\bm{p}$.

The process, which constructs the configuration $\bm{D} \in \{0,1\}^{N \times M}$ by sequentially filling each sample, is defined in Algorithm~\ref{alg:lpm-init}.

\begin{algorithm}[htb!] 
\caption{Sampling-based initialization of tactical configuration}
\begin{algorithmic}[1]
\State \textbf{Input:}
Population and sample size $N,n$.
Fixed-size sample generator $f$.
Tactical configuration parameters $M,c$.
\State \textbf{Compute:}
Initial budget $\bm{b}_0 = \bm{1}c$.
\For{$k = 1$ to $M$}
    \State Calculate the remaining budget for the units: 
    \State $\bm{b}_k \leftarrow \bm{b}_{k-1} - \bm{d}_{k-1}$, where $\bm{d}_0 = \bm{0}$.
    \State Set inclusion probabilities:
    \State $\bm{p}_k \leftarrow \bm{b}_k / (M-k+1)$.
    \State Select a sample using $f$:
    \State $\bm{d}_k \leftarrow f(\bm{p}_k)$.
\EndFor
\State \Return Tactical configuration $\bm{D} = [\bm{d}_1 | \bm{d}_2 | \cdots | \bm{d}_M]$.
\end{algorithmic} \label{alg:lpm-init}
\end{algorithm}

The budget $\bm{b}_k \in \{0,\ldots,M-k+1\}^N$ gives the number of additional inclusions needed for each unit.
As $\bm{b}_k$ sums to $n(M-k+1)$, $\bm{p}_k$ sums to $n$.
Thus, Algorithm~\ref{alg:lpm-init} yields a tactical configuration.

Using the LPM as the sample generator $f$ guarantees a well-spread sample of exactly $n$ units at each step. This approach yields an initial configuration in which each sample is generated by a spatially balanced procedure. For the simulated annealing process, this "warm start" may reduce the total number of iterations required to obtain a near-optimal design.

\subsection{Optimization algorithm}\label{sec:algorithm}

Algorithm~\ref{alg:optimization} presents the pseudo-code for finding an optimized configuration $\bm{D}^\circ$, starting from an initial state $\bm{D}^{(0)}$. The optimization operates over $\mathcal{P}$ using elementary transitions defined by pairwise exchanges. Specifically, two samples (columns) $\bm{d}_k, \bm{d}_l$ with $k \neq l$ are selected and a swap between units $i$ and $j$, where $d_{ik} = 1$ and $d_{jl} = 1$, is admissible if $d_{il} = 0$ and $d_{jk} = 0$. See Figure~\ref{fig:vector_swap} for an illustration. At each iteration, two random units are selected from two samples. If a swap is admissible, the units are moved according to the simulated annealing procedure \citep[see e.g.][]{Kea83}.

\begin{figure}[htb!]
    \centering
    \includegraphics[width=0.6\linewidth]{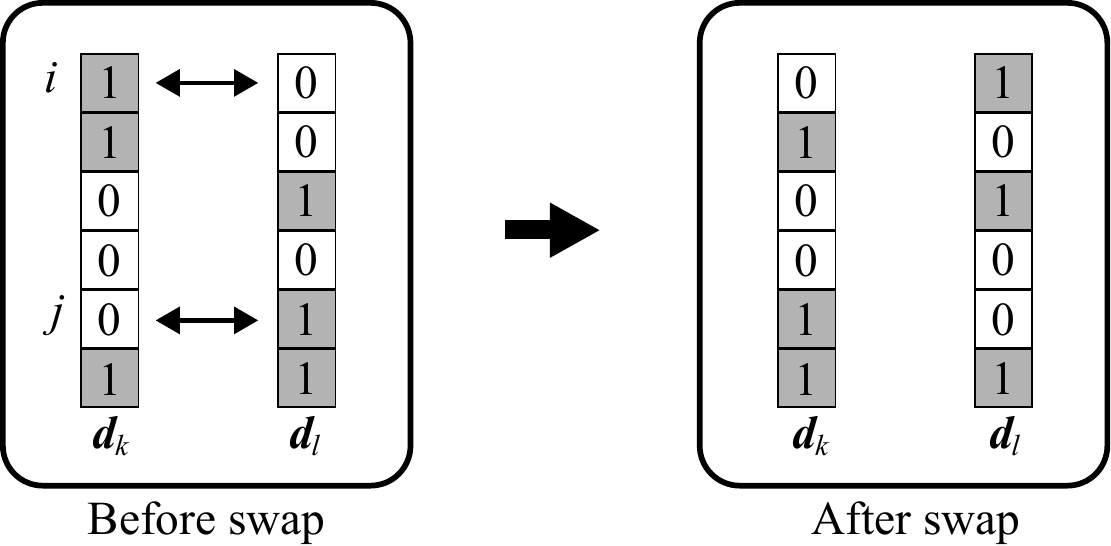}
    \caption{Illustration of an admissible swap between a unit $i$ in $\bm{d}_k$ ($d_{ik}=1$) and a unit $j$ in $\bm{d}_l$  ($d_{jl}=1$). The swap is admissible as $d_{il}=0$ and $d_{jk}=0$ (before the swap).}
    \label{fig:vector_swap}
\end{figure}

\begin{algorithm}[htb!] 
\caption{Optimization through simulated annealing}
\begin{algorithmic}[1] 
\State \textbf{Input:}
    \State Initial tactical configuration $\bm{D}^{(0)}$ of size $N \times M$.
    \State Iterations $R$, initial temperature $T^{(0)}>0$, and cooling rate $0<\alpha<1$.
\State \textbf{Compute:} 
    Expected energy distance $\bar{\mathcal{E}}^{(0)} = \bar{\mathcal{E}}(\bm{D}^{(0)})$.
\State \textbf{Initialize:} 
    Best state pair $\bm{D}^{\circ} \leftarrow \bm{D}^{(0)}$; and
    $\bar{\mathcal{E}}^\circ \leftarrow \bar{\mathcal{E}}^{(0)}$.
\For{$r = 1$ to $R$}
    \State $\bm{D}^{(r)} \leftarrow \bm{D}^{(r-1)}$, \quad $\bar{\mathcal{E}}^{(r)} \leftarrow \bar{\mathcal{E}}^{(r-1)}$.
    \State Randomly select sample indices $k,l \in \{1,\ldots,M\}$ uniformly without replacement.
    \State Randomly select units $i,j$ uniformly from the set of units in samples $k,l$ respectively, 
    i.e. such that $d^{(r)}_{ik} = d^{(r)}_{jl} = 1$.

    \If{$d^{(r)}_{il} = d^{(r)}_{jk} = 0$}
        \State Swap the units
        \State $d^{(r)}_{ik} \leftarrow 0; d^{(r)}_{il} \leftarrow1$;
        \State $d^{(r)}_{jl} \leftarrow 0; d^{(r)}_{jk} \leftarrow1$;
        \State and compute the expected energy distance $\bar{\mathcal{E}}^{(r)}$.

        \If{$\bar{\mathcal{E}}^{(r)} < \bar{\mathcal{E}}^{\circ}$}
            \State New best state $\bm{D}^\circ \leftarrow \bm{D}^{(r)}$, $\bar{\mathcal{E}}^\circ \leftarrow \bar{\mathcal{E}}^{(r)}$.
            
        \ElsIf{$\bar{\mathcal{E}}^{(r)} \geq \bar{\mathcal{E}}^{(r-1)}$ \textbf{and} $U(0,1) \geq \exp(-(\bar{\mathcal{E}}^{(r)} - \bar{\mathcal{E}}^{(r-1)}) / T^{(r-1)})$}
        \State Reject swap and reset $\bm{D}^{(r)} \leftarrow \bm{D}^{(r-1)}$, $\bar{\mathcal{E}}^{(r)} \leftarrow \bar{\mathcal{E}}^{(r-1)}$.
        \EndIf
    \EndIf
    \State Update temperature $T^{(r)} \leftarrow \alpha T^{(r-1)}$.    
\EndFor
\State \Return Optimized tactical configuration $\bm{D}^\circ$.
\end{algorithmic} \label{alg:optimization}
\end{algorithm}

The algorithm is terminated after a pre-determined number of iterations $R$.
The temperature $T$ controls the probability of accepting a sub-optimal swap, which may reduce the risk of becoming stuck in a local minimum.
As the algorithm progresses, the temperature may be reduced by the cooling rate $\alpha$, as sub-optimal swaps may be less preferable in the later stages.

As a tactical configuration is an $N\times M$ binary incidence matrix with constant sums of rows and columns, and an admissible swap is exactly a $2 \times2$ interchange, the margins are preserved during admissible swaps. \citet{Ryser57} established that the class of such matrices is connected under elementary interchanges, which corresponds to the pairwise swaps defined in our algorithm. This connectivity property ensures that local stochastic optimization methods, such as simulated annealing, can explore the entire admissible design space.

In Appendix~\ref{app:update}, we show that the cost of calculating the expected energy distance in each step is limited by $O(n)$ and provide a remark on parallel updates.
This can be contrasted with simple random sampling (SRS), which draws a single sample in $O(n)$.
In Appendix~\ref{app:large}, we provide some strategies that can be helpful in dealing with large populations. 

A random sample $\bm{\delta}$ is obtained by drawing a sample uniformly at random from the optimized configuration $\bm{D}^{\circ}$.
An efficient implementation of Algorithm~\ref{alg:optimization} is available in the \texttt{rsamplr} R package \citep{P25}.

\section{Simulation} \label{sec:simulation}
In order to evaluate the performance of DBD-TC, we performed a simulation study.
A synthetic population $U$ of size $N=1000$ was generated, with a set of $p=10$ auxiliary variables.
The auxiliaries $\bx_i$ were drawn from a standard uniform distribution $\bm{U}_p(0,1)$, independently for each unit $i \in U$.
Additionally, three target variables $\by$ were generated as functions of the auxiliaries, using a linear relationship, a quadratic relationship, and a sinusoidal relationship, standardized to have mean 0 and variance 1:
\begin{align*}
    y_{1} &= \sqrt{\frac{12}{5}} \left( x_{1} + 2x_{2} - \frac{3}{2} \right) ,\\
    y_{2} &= \sqrt{\frac{180}{79}} \left( x_{1} + 8(x_{2} - 0.5)^2 - \frac{7}{6} \right) ,\\
    y_{3} &= \sqrt{\frac{12}{7}} \left( x_{1} + \sin (2 \pi x_{2}) - \frac{1}{2} \right) .
\end{align*}

The results obtained from DBD-TC were compared with Circular DBD, as well as with two of the leading methods for balanced and spatially balanced sampling: the local cube (LCUBE) method \citep{GrafstromTille2013} and the LPM.
LCUBE is an extension of the cube method, inheriting its balancing properties while also providing spatial balance similar to LPM.
Additionally, SRS served as an uninformative sampling baseline.
Four different evaluation metrics were used; the energy distance ($\mathcal{E}$) given by \eqref{eq:ed}, the mean spatial balance (SB) \citep{StevensOlsen2004}, the mean local balance (LB) \citep{PG24}, and balance deviation (BD).
The first three quantify different aspects of spatial or distributional discrepancy. The BD measures the Euclidean distance between the HT estimates and the known auxiliary totals. It is defined as
\[
BD(\bm{\delta})=\|\hat{\bX}_{\bm{\delta}}-\bX\|,
\]
where 
\[
\hat{\bX}_{\bm{\delta}}=\sum_{i \in U}\frac{\bx_i}{\pi_i} \delta_i
,\quad
\bX=\sum_{i\in U}\bx_i
.
\]
The lower the average score, the better the design performs.

In Figure~\ref{fig:rate}, we explore the rate at which the expected energy distance decays as a function of the number of iterations in the DBD algorithms, using a sample size of $n=50$ and the first $p=2$ auxiliary variables.
The results show that DBD-TC reaches a mean energy distance seemingly unattainable with Circular DBD after about $10^6$ iterations.

\begin{figure}[htb!]
    \centering
    \includegraphics[width=1\linewidth]{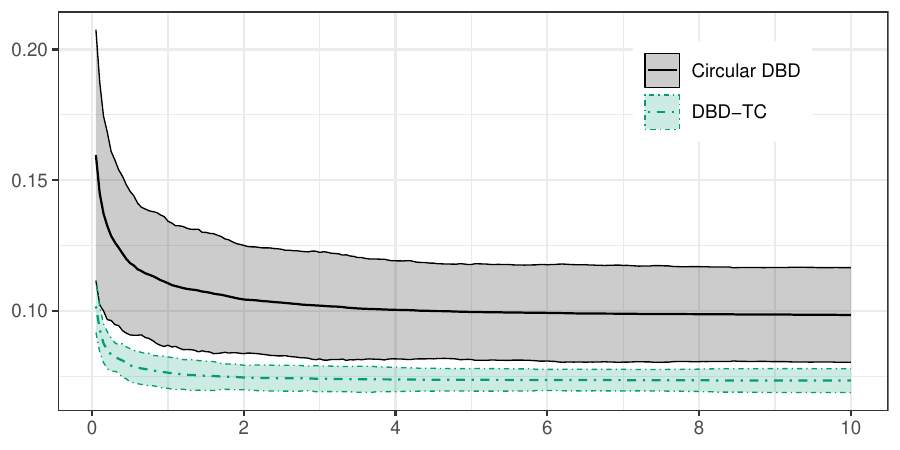}
    \caption{
        Expected energy distance (scaled $\times 100$) $\pm 2$ standard deviations of Circular DBD (solid gray) and DBD-TC (dashed green) as the number of iterations ($10^6$ in total) increases, for sample size $n=50$ and $p=2$ auxiliary variables.
    }
    \label{fig:rate}
\end{figure}

The competing sampling methods were evaluated using two sets of auxiliaries with dimensions $p \in \{2,10\}$, and two different sample sizes $n \in \{50, 200\}$ using equal inclusion probabilities.

DBD-TC and Circular DBD were evaluated using their respective complete designs obtained after $10^7$ iterations, whereas a Monte Carlo simulation was performed for the other competing methods, using $10,000$ iterations.

In addition to evaluating distributional fit, the methods were evaluated with respect to the target variables and their unbiased HT-estimator. Since the HT estimator is unbiased, its MSE equals its design variance. We estimate the MSE from the simulation, while for each realized sample the variance is estimated by $\hat{V}$, using the local mean variance estimator \citep{GrafstromSchelin2014} for the spatially balanced methods and the HT variance estimator in the case of SRS. The variance estimates were used to construct approximate 95\% confidence intervals, whose coverage was then evaluated.

The results, presented in Table~\ref{tab:sb_metrics}, show DBD-TC to be a top performer over all metrics and settings, implying an improved distributional fit that remains as the sample size or auxiliary space grows.

\begin{table}[htb!]
\centering
\caption{
Comparison of distributional fit through the empirical means of three measures of spatial balance ($\mathcal{E}$, SB, LB) and one measure of balance (BD), for different sample sizes $n$ and number of auxiliary variables $p$.
The energy distance $\mathcal{E}$ is presented scaled $\times 100$. 
Lower values indicate better performance.
}
\label{tab:sb_metrics}
\vspace{0.2cm}
\begin{tabular}{rl rrrr rrrr}
&& \multicolumn{4}{c}{$p=2$} & \multicolumn{4}{c}{$p=10$} \\
\cmidrule(lr){3-6}\cmidrule(lr){7-10}
$n$ & Method & 
$\bar{\mathcal{E}}$ & $\overline{\text{SB}}$ & $\overline{\text{LB}}$ & $\overline{\text{BD}}$ & 
$\bar{\mathcal{E}}$ & $\overline{\text{SB}}$ & $\overline{\text{LB}}$ & $\overline{\text{BD}}$ \\
\midrule 50 
&              SRS &   0.985 &   0.333 &   0.145 &  49.479 &   2.411 &   0.320 &   0.268 & 122.897 \\
&              LPM &   0.152 &   0.088 &   0.077 &  10.525 &   1.448 &   0.251 &   0.251 &  74.635 \\
&            LCUBE &   0.126 &   0.080 &   0.075 &   8.080 &   1.039 &   0.244 &   0.249 &  25.568 \\
&     Circular DBD &   0.099 &   0.061 &   0.065 &   5.069 &   0.954 &   0.246 &   0.249 &  23.460 \\
&           DBD-TC &   0.074 &   0.042 &   0.056 &   0.994 &   0.855 &   0.217 &   0.241 &  11.837 \\
\midrule 200 
&              SRS &   0.206 &   0.354 &   0.074 &  22.539 &   0.506 &   0.320 &   0.135 &  56.198 \\
&              LPM &   0.016 &   0.129 &   0.044 &   2.710 &   0.257 &   0.237 &   0.123 &  28.842 \\
&            LCUBE &   0.013 &   0.132 &   0.045 &   2.002 &   0.187 &   0.258 &   0.126 &   6.676 \\
&     Circular DBD &   0.011 &   0.114 &   0.041 &   0.848 &   0.168 &   0.251 &   0.125 &   5.880 \\
&           DBD-TC &   0.008 &   0.107 &   0.040 &   0.187 &   0.150 &   0.246 &   0.123 &   2.422 \\
\bottomrule
\end{tabular}
\end{table}

In Table~\ref{tab:variance}, where we present the results for the three target variables, we see that DBD-TC is very efficient in terms of yielding low variance of the estimators. The results for $p=10$ also illustrate how the effectiveness of all methods using auxiliary variables decreases when the designs are over-saturated with auxiliaries unrelated to the target variables.

\begin{table}[htb!]
\centering
\caption{
Comparison of the square roots of MSE (RMSE) and mean $\hat{V}$, alongside 95\% CI coverage (Cov), for three target variables ($y$) as functions of $x_1,x_2$, for different sample sizes $n$ and number of auxiliary variables $p$.
}
\label{tab:variance}
\vspace{0.2cm}
\begin{tabular}{rrl rrr rrr}
&&& \multicolumn{3}{c}{$p=2$} & \multicolumn{3}{c}{$p=10$} \\
\cmidrule(lr){4-6}\cmidrule(lr){7-9}
$y$ & $n$ & Method & RMSE & $\sqrt{\bar{\hat{V}}}$ & Cov. & RMSE & $\sqrt{\bar{\hat{V}}}$ & Cov. \\
\midrule\multirow{10}{*}{\rotatebox[origin=c]{90}{$y_1$, linear}}
&50 
 &              SRS & 137.628 & 138.156 &   0.946 & 138.715 & 138.257 &   0.944 \\
&&              LPM &  28.915 &  25.568 &   0.914 &  84.058 &  84.844 &   0.949 \\
&&            LCUBE &  22.234 &  25.596 &   0.975 &  19.012 &  84.671 &   1.000 \\
&&     Circular DBD &  14.775 &  25.332 &   1.000 &  27.590 &  84.521 &   1.000 \\
&&           DBD-TC &   2.137 &  24.683 &   1.000 &  16.805 &  85.562 &   1.000 \\
\cmidrule(lr){2-9}
&200 
 &              SRS &  63.596 &  63.370 &   0.946 &  63.838 &  63.395 &   0.949 \\
&&              LPM &   7.460 &   5.316 &   0.837 &  32.451 &  33.203 &   0.956 \\
&&            LCUBE &   5.585 &   5.313 &   0.938 &   4.912 &  32.946 &   1.000 \\
&&     Circular DBD &   2.408 &   5.213 &   0.999 &   6.735 &  32.848 &   1.000 \\
&&           DBD-TC &   0.416 &   5.011 &   1.000 &   2.858 &  33.182 &   1.000 \\
\midrule\multirow{10}{*}{\rotatebox[origin=c]{90}{$y_2$, quadratic}}
&50 
 &              SRS & 140.029 & 140.461 &   0.944 & 141.215 & 140.446 &   0.943 \\
&&              LPM &  47.055 &  61.673 &   0.987 & 117.770 & 124.865 &   0.955 \\
&&            LCUBE &  32.538 &  62.279 &   1.000 & 103.171 & 125.194 &   0.979 \\
&&     Circular DBD &  23.300 &  62.600 &   1.000 &  95.595 & 125.324 &   0.973 \\
&&           DBD-TC &   6.790 &  62.224 &   1.000 &  65.162 & 126.901 &   1.000 \\
\cmidrule(lr){2-9}
&200 
 &              SRS &  64.225 &  64.429 &   0.948 &  64.954 &  64.447 &   0.947 \\
&&              LPM &  12.839 &  12.627 &   0.943 &  50.400 &  55.797 &   0.967 \\
&&            LCUBE &   8.237 &  12.733 &   0.997 &  41.524 &  55.709 &   0.991 \\
&&     Circular DBD &   5.264 &  12.654 &   1.000 &  38.691 &  54.841 &   1.000 \\
&&           DBD-TC &   1.006 &  12.437 &   1.000 &   8.774 &  56.977 &   1.000 \\
\midrule\multirow{10}{*}{\rotatebox[origin=c]{90}{$y_3$, sinusoidal}}
&50 
 &              SRS & 136.534 & 136.866 &   0.945 & 137.133 & 136.843 &   0.940 \\
&&              LPM &  36.518 &  44.959 &   0.983 &  97.850 & 105.037 &   0.959 \\
&&            LCUBE &  28.019 &  45.301 &   0.999 &  79.203 & 104.703 &   0.987 \\
&&     Circular DBD &  20.971 &  45.195 &   1.000 &  71.406 & 105.412 &   0.996 \\
&&           DBD-TC &   5.749 &  44.940 &   1.000 &  73.735 & 105.226 &   1.000 \\
\cmidrule(lr){2-9}
&200 
 &              SRS &  62.165 &  62.822 &   0.950 &  62.260 &  62.809 &   0.952 \\
&&              LPM &  10.310 &  10.063 &   0.944 &  39.672 &  44.978 &   0.972 \\
&&            LCUBE &   7.116 &  10.112 &   0.994 &  35.075 &  44.660 &   0.988 \\
&&     Circular DBD &   5.005 &  10.104 &   1.000 &  51.272 &  44.449 &   0.852 \\
&&           DBD-TC &   1.446 &   9.972 &   1.000 &  22.988 &  44.687 &   1.000 \\
\bottomrule
\end{tabular}
\end{table}

On a standard office laptop, constructing the entire DBD-TC design using $10^7$ iterations takes about 6 seconds for $n=50$ and $p=2$, approximately half the time required by Circular DBD. By comparison, SRS and LPM draw a single sample in milliseconds.

\subsection{The Meuse dataset}
The Meuse dataset consists of 164 locations in a river flood plain \citep{gstat}, out of which two locations were removed due to missing values.
At each location, a number of topsoil heavy metal concentrations have been measured, together with other environmental variables.
DBD-TC and the competing methods were evaluated by performing a simulation, in which the respective algorithms used five standardized auxiliary variables: coordinates ($x, y$), elevation (elev), organic matter (org) and copper concentration (Cu).
The target variables for the evaluation were the standardized concentrations of zinc (Zn), lead (Pb), and cadmium (Cd).

DBD-TC and Circular DBD were evaluated using their respective complete designs obtained after $10^7$ iterations, whereas a Monte Carlo simulation was performed for the other methods, using $10,000$ iterations.

\begin{table}[htb!]
\centering
\caption{
Comparison of the square roots of MSE (RMSE) and mean $\hat{V}$, alongside 95\% CI coverage (Cov), for the target variables (Zn, Pb, Cd) and auxiliaries (Cu, elev, org) of the Meuse dataset, alongside evaluation of distributional fit through the empirical means of measures of spatial balance ($\mathcal{E}$, SB, LB), for sample size $n=20$.
Lower values indicate better performance.
}
\label{tab:meuse_results}
\small
\begin{tabular}{ll rrrrrr rrr}
\toprule
&Method & Zn & Pb & Cd & Cu & elev & org
& $\bar{\mathcal{E}}$ & $\overline{\text{SB}}$ & $\overline{\text{LB}}$ \\
\midrule\multirow{5}{*}{\rotatebox[origin=c]{90}{RMSE}}
&              SRS & 33.94 & 33.73 & 33.75 & 33.75 & 34.27 & 34.11 &  0.13 &  0.41 &  0.36 \\
&              LPM & 20.85 & 20.57 & 17.51 & 14.61 &  9.56 & 14.36 &  0.04 &  0.17 &  0.27 \\
&            LCUBE & 19.83 & 19.09 & 17.30 & 11.72 & 10.53 & 11.42 &  0.04 &  0.19 &  0.28 \\
&     Circular DBD & 11.06 & 17.81 & 14.82 &  7.99 &  6.18 &  6.72 &  0.03 &  0.17 &  0.27 \\
&           DBD-TC & 16.27 & 16.27 & 12.69 &  2.88 &  2.55 &  2.16 &  0.03 &  0.15 &  0.25 \\
\midrule\multirow{5}{*}{\rotatebox[origin=c]{90}{$\sqrt{\bar{\hat{V}}}$}}
&              SRS & 33.89 & 33.92 & 33.90 & 33.93 & 33.87 & 33.96 &&& \\
&              LPM & 24.99 & 26.57 & 20.41 & 19.32 & 28.17 & 20.67 &&& \\
&            LCUBE & 24.94 & 26.80 & 20.75 & 19.73 & 27.71 & 20.93 &&& \\
&     Circular DBD & 25.95 & 27.15 & 20.84 & 20.10 & 28.08 & 21.14 &&& \\
&           DBD-TC & 25.31 & 27.11 & 21.09 & 20.30 & 28.58 & 21.53 &&& \\
\midrule\multirow{5}{*}{\rotatebox[origin=c]{90}{Coverage}}
&              SRS &  0.91 &  0.91 &  0.91 &  0.92 &  0.84 &  0.93 &&& \\
&              LPM &  0.96 &  0.97 &  0.96 &  0.98 &  0.99 &  0.99 &&& \\
&            LCUBE &  0.97 &  0.98 &  0.97 &  1.00 &  0.96 &  1.00 &&& \\
&     Circular DBD &  1.00 &  1.00 &  1.00 &  1.00 &  0.99 &  1.00 &&& \\
&           DBD-TC &  1.00 &  1.00 &  1.00 &  1.00 &  1.00 &  1.00 &&& \\
\bottomrule
\end{tabular}
\end{table}

From the Meuse results in Table~\ref{tab:meuse_results}, we see that DBD-TC achieves the best distributional fit in the auxiliary space, with the lowest mean energy distance and the best spatial-balance measures among the competing designs. This improved auxiliary-space representation also provides efficiency gains for the target variables. 

\clearpage

\section{Final remarks} \label{sec:remarks}

This paper extends the DBD framework by replacing the circular ordering construction of \citet{GrafstromPrentiusDBD} with optimization over tactical configurations. 

A particularly striking result is that DBD-TC achieves superior balance despite restricting optimization to tactical configurations of minimum size. For more than 20 years, the cube method \citep{DevilleTille2004} has been the gold standard for balanced sampling focusing on the realization of samples that meet linear constraints. In all simulation settings, the balance deviation of DBD-TC is considerably smaller than that of LCUBE, which inherits the same level of balance as the cube method \citep{GrafstromTille2013}. We believe that the main reason lies in the difference between constructive sampling and support optimization. The cube method, LCUBE, and LPM construct one realization through a sequence of sampling decisions, whereas DBD-TC jointly optimizes a complete collection of samples. The results suggest that this additional flexibility improves not only distributional balance, but also linear balance and spatial spread.

The improved design properties translate into substantial gains in estimation efficiency. At the same time, the local mean variance estimates and hence the widths of the resulting confidence intervals are broadly comparable to those obtained under LPM, LCUBE, and Circular DBD.
Therefore, the higher coverage achieved by DBD-TC is not due to wider intervals.
Rather, the local mean variance estimator does not fully reflect the sometimes very large reduction in the actual sampling variance achieved by DBD-TC. Among the methods considered here, DBD-TC achieves a lower RMSE while maintaining at least a comparable interval width. This makes it a preferable choice when its equal-probability setting is appropriate.

Support optimization requires more computation than constructing a single sample.
The proposed algorithm evaluates an admissible exchange in $O(n)$ time and can be extended to allow for parallel updates. In the simulation settings considered here, complete DBD-TC designs were obtained in a few seconds. A remaining limitation is that the present framework is restricted to equal inclusion probabilities, unlike the cube method, which naturally accommodates unequal probabilities. 

\section*{Acknowledgments}

The authors thank the associate editor and two anonymous referees for their comments, which helped to improve the article.

\iffalse\section{Bibliography}\fi

\appendix

\section{Efficient update of the objective function}\label{app:update}

\subsection{Incremental updates}
Consider a current tactical configuration $\bm{D}$ and an admissible swap between units $u,v$ in samples $\bm{d}_a, \bm{d}_b$, where $d_{ua}=d_{vb}=1$, creating the new configuration $\bm{D}'$.

From \eqref{eq:energy-sample} we get the total energy of the configuration as
\[
\mathcal{E}(\bm{D}) =
\sum_{k=1}^M \mathcal{E}(F_{\bm{d}_k},F_U)
=
\frac{c}{n} \sum_{i \in U} \Phi_i
-
\sum_{k=1}^M \sum_{i \in U} \sum_{j \in U} \frac{d_{ik}}{n}\frac{d_{jk}}{n} \|\bx_i-\bx_j\|
.
\]
The change in total energy is
\[
\Delta \mathcal{E} 
= \mathcal{E}(\bm{D}') - \mathcal{E}(\bm{D})
=
\frac{1}{n^2} \sum_{i \in U} \sum_{j \in U}
\left(
d_{ia} d_{ja} - d_{ia}' d_{ja}' 
+
d_{ib} d_{jb} - d_{ib}' d_{jb}' 
\right) 
\|\bx_i-\bx_j\|
.
\]
As $d_{ia} d_{ja} - d_{ia}' d_{ja}' = 0$ for any pair $i,j$ not including $u,v$, and vice versa for $b$, we get
\begin{multline*}
\Delta \mathcal{E}
=
\frac{2}{n^2}
\sum_{i \in U} 
\left(
( d_{ia} - d_{ib}' )
\|\bx_i-\bx_u\| 
+
( d_{ib} - d_{ia}' )
\|\bx_i-\bx_v\|
\right)
=\\=
\frac{2}{n^2}
\sum_{i \in U \setminus \{u,v\}} 
( d_{ia} - d_{ib} )
\left(
\|\bx_i-\bx_u\| 
-
\|\bx_i-\bx_v\|
\right)
,
\end{multline*}
using the fact that $d_{ia} d_{ua} - d_{ia}' d_{ua}' = d_{ia}$ and $d_{ia} d_{va} - d_{ia}' d_{va}' = -d_{ia}'$, and analogously for sample $b$. The final equality follows as for $i\notin\{u,v\}$, we have $d_{ia} = d_{ia}'$ and $d_{ib} = d_{ib}'$. As only units in exactly one of the samples $\bm{d}_a$ or $\bm{d}_b$ affect the change in energy, the total computational cost to evaluate $\Delta{\mathcal{E}}$ is $O(n)$.

\subsection{Parallel updates}
Since the design objective is defined as a sum of sample energies, any swaps involving disjoint pairs of samples are non-interacting in the objective. Consequently, up to $\lfloor M/2 \rfloor$ updates can be performed in parallel within a single iteration, potentially providing substantial computational speedups for designs with large support sizes. 

\section{Sampling of large populations} \label{app:large}
\subsection{Stratification}
For extremely large populations, the DBD-TC framework can be applied within strata defined by auxiliary variables, as in the Block-DBD approach \citep{GrafstromPrentiusDBD}. 
Each stratum is treated independently, allowing local minimum-size configurations, $O(n)$ incremental updates, and parallel swaps. The stratified samples can then be combined to form a global design that maintains distributional representativeness, making the method suitable for very large-scale applications.

\subsection{LPM-based compression for large populations}

When the population size $N$ is very large, the minimum configuration size
$M = N/\gcd(N,n)$ implied by Proposition~\ref{prop:min_support} may still be computationally prohibitive. In such cases, the DBD-TC construction can be preceded by a randomized population compression step that reduces the effective configuration size in a controlled manner. Specifically, one may choose a positive integer $M^\ast$ that is smaller than or equal to $\lfloor N/n \rfloor$, and define a reduced population size $N^\ast = M^\ast n$. Then a subset $U^\ast \subset U$ with $|U^\ast| = N^\ast$ is selected using the local pivotal method (LPM) with equal inclusion probabilities. The empirical distribution $F_{U^\ast}$ provides an approximation of the population distribution $F_U$, and the approximation error becomes small when $N^\ast\gg n$. Conditionally on $U^\ast$, the DBD-TC construction is applied with population size $N^\ast$ and sample size $n$. This procedure provides correct unconditional inclusion probabilities and introduces a tunable accuracy–complexity tradeoff that removes arithmetic discontinuities and makes optimization feasible for very large populations. However, it only produces a conditional design.

\begin{example}[LPM-based compression for large populations]
\label{ex:lpm-reduction}
Consider a population of size $N = 10^6$ and a sample size $n=51$. 
In this case, $\gcd(N,n) = 1$, which implies a minimum configuration size of $M = 10^6$. To reduce the configuration while preserving exact inclusion probabilities, we choose a smaller target size, e.g., $M^\ast = \lfloor N/n \rfloor = 19{,}607$. This corresponds to a reduced population size of $N^\ast = M^\ast n = 999{,}957$. 

By selecting the $N^\ast$ units for $U^\ast$ from $U$ via the local pivotal method, we are effectively rounding the weights $W_i = M^\ast (n/N) $ to integer frequencies $n_i \in \{0,1\}$ such that $\sum n_i = N^\ast$. A DBD-TC design is then constructed by partitioning these $999{,}957$ units into $19{,}607$ samples. Unconditionally, each unit $i \in U$ is included in a random sample $S$ with probability:
\[
\Pr(i \in S) = \Pr(i \in U^\ast)\Pr(i\in S \mid i\in U^\ast) = \frac{M^\ast n}{N}\frac{1}{M^\ast} = \frac{n}{N} = \frac{51}{10^6}.
\] 
This produces an excellent approximation to the distributionally balanced design that would have been obtained on the full population, while avoiding the arithmetic discontinuity of the $\gcd$ function and dramatically reducing the computational burden.
\end{example}

\end{document}